\documentclass[leqno]{amsart}

\usepackage{tsarxiv}

\begin{document}

\title{Shot-Noise Processes in Finance}

\author{Thorsten Schmidt}
\address{University Freiburg, Dep. of Mathematical Stochastics, Eckerstr. 1, 79104 Freiburg, Germany. Email: Thorsten.Schmidt@stochastik.uni-freiburg.de}

\date{\today}

\maketitle

\begin{abstract}
Shot-Noise processes constitute a useful tool in various areas, in particular in finance. They allow to
model abrupt changes in a more flexible way than processes with jumps and hence are an ideal tool for modelling stock prices,  credit portfolio risk, systemic risk, or electricity markets. Here we consider a general formulation of shot-noise processes, in particular time-inhomogeneous shot-noise processes. This flexible class allows to obtain the Fourier transforms in explicit form and is highly tractable. We prove that Markovianity is equivalent to exponential decay of the noise function.  Moreover, we study the relation to semimartingales and equivalent measure changes which are essential for the financial application. In particular we derive a drift condition which guarantees absence of arbitrage. Examples include the minimal martingale measure and the Esscher measure. 
\end{abstract}


\section{Introduction}

Shot-noise processes constitute a well-known tool for modelling sudden changes (\emph{shots}), followed by a typical following pattern (\emph{noise}). In this regard, they are more flexible than other approaches simply utilizing jumps and this led to many applications in physics, biology and, with an increasing interest, in finance. 
Quite remarkably, shot-noise effects were already introduced in the early 20th century, see \cite{Schottky1918,Campbell1909b,Campbell1909a}, sometimes also referred to as Schottky-noise. First fundamental treatments were only developed many years later with\footnote{The works stem from different authors, Stephen Oswald Rice and John Rice.} \cite{Rice1944,Rice1945} and \cite{Rice1977}.
Applications of shot-noise processes also arise in insurance mathematics, marketing, and even astronomy - see the survey article \cite{BondessonSN}.
The first appearances in a finance context seem to be \cite{Samorodnitsky1995,Chobanov1999} while in insurance mathematics this class of processes were studied even earlier, see \cite{klüppelberg2003} for literature in this regard.

In a general form, denote by $0 < T_1 < T_2 < \dots$ the arrival times of the shots, and by $(H(.,T): T \in \R_{\ge 0})$ a family of stochastic processes representing the noises, then a \emph{shot-noise process} $S$ is given by the superposition
\begin{align} \label{SN:1} 
S_t = \sum_{i \ge 1} \ind{t \le T_i} H(t,T_i), \quad t \ge 0; 
\end{align}
an example at this level of generality can be found in \cite{SchmidtStute07}.
Of course, absolute convergence of the sum needs to be justified, typically by making assumptions on the arrival times together with suitable restrictions on the noise processes. For the consideration of stationarity, the process is often extended to the full real line.

At this level of generality, shot-noise processes extend compound Poisson processes significantly and neither need to be Markovian nor semimartingales.
While the definition in \eqref{SN:1} is very general, more restrictions will be needed to guarantee a higher level of tractability. In this paper we will focus on shot-noise processes which are semimartingales. To the best of our knowledge all articles, including \cite{Rice1977} and many others, assume that the noises are i.i.d.~and independent from the arrival times of the shots. The most common assumption even leads to a piecewise deterministic Markov process: this is the case, if the noise processes are given by
 $ H(t,T_i) = U_i e^{-a(t-T_i)}$ with i.i.d.~$(U_i)_{i \ge 1}$, independent from $(T_i)_{i \ge 1}$, and $a \in \R$.
We will show later that this is essentially the only example where Markovianity is achieved. 
More general cases allow for different decay, as for example a power-law decay, see  e.g.~\cite{Lowen1990}, or do not assume a multiplicative structure for the jump heights $(U_i)$. These cases can be summarized under the assumption that 
\begin{align}\label{snpaper}
H(t,T_i) = G(t-T_i,U_i), \qquad t \ge 0,\ i \ge 1, 
\end{align}
with some general random variables $(U_i)$ and a suitable (deterministic) function $G$.

The obtained class of processes is surprisingly tractable, and the reason for this is that the Fourier and Laplace transforms of $S$ are available in explicit form, depending on the considered level of generality. 
Due to linearity of the integral, the class is stable under integration,  a property  shared by affine processes and of high importance for applications in interest rate markets and credit risk, see \cite{GasparSchmidt10}.

A branch of literature considers limits of shot-noise processes when the intensity of the shot arrivals increases and shows, interestingly, that limits of this class of processes have fractional character, see \cite{Lane1984,Lowen1990,KlueppelbergKuehn2004}, and the early studies in insurance mathematics. 

The application of shot-noise processes to the modelling of consumer behaviour has been suggested in \cite{KopperschmidtStute2009,KopperschmidtStute2013}, wherein also the necessary statistical tools have been developed. The key in this approach is that i.i.d.\ shot-noise processes are at hand which allows a good access to statistical methodologies.

In the financial and insurance community they have been typically
used to efficiently model shock effects, see for example \cite{DassiosJang2003},  \cite{AlbrecherAsmussen2006}, \cite{SchmidtStute07}, \cite{AltmannSchmidtStute08},  \cite{JangHerbertssonSchmidt}, \cite{SchererSchmidSchmidt12},  and references therein. Besides this, in \cite{MorenoSerranoStute2011} an estimation procedure in a special class of shot-noise processes utilizing the generalized method of moments (GMM) is developed.

The paper is organized as follows: in section \ref{sec:setup} we introduce a suitably general formulation of shot-noise processes and derive their conditional characteristic function. Moreover, we study the connection to semimartingales and Markov processes. Proposition \ref{prop:Markovian} proves that exponential decay is equivalent to Markovianity of the shot-noise process. In Section \ref{sec:FM} we propose a model for stocks having a shot-noise component. After the study of equivalent and absolutely continuous measure changes we obtain a drift condition implying absence of arbitrage and give an example where independence and stationarity of increments holds under the objective and the equivalent martingale measure.

\section{Shot-noise processes}\label{sec:setup}
Our focus will lie on shot-noise processes satisfying \eqref{snpaper} and the detailed study of this flexible class. 
Consider a  filtered probability space $(\Omega,\cF,\bbF,\P)$ where the filtration $\bbF=(\cF_t)_{t \ge 0}$ satisfies the usual conditions, i.e.\ $\bbF$ is right-continuous and $A\subset B \in \cF$ with $\P(B)=0$ implies $A \in \cF_0$. By $\cO$ and $\cP$ we denote the optional, respectively predictable, $\sigma$-fields, generated by the c\`adl\`ag, respectively c\`ag, processes.

We will allow for a  marked point process as driver\footnote{We consider here for simplicity $\R^d$ as mark space, while $\R^d$ can be replaced by a general Lusin space, see \cite{BKR} in this regard.}, generalizing previous literature. In this regard, consider a sequence of increasing stopping times $0<T_1<T_2<\dots$ and a sequence of $d$-dimensional random variables $U_1,U_2,\dots$. The double sequence $Z=(T_i,U_i, i \ge 1)$ is called \emph{marked point process}. Such processes are well-studied in the literature and we refer to \cite{Bremaud1981} for further details and references. We consider one-dimensional shot-noise processes only, a generalization to more (but finitely many) dimensions is straightforward; for the more general case see e.g.\ \cite{BassanBona87} for shot-noise random fields.

\begin{definition}   \label{def:shotnoise}
If $Z=(T_i,U_i, i \ge 1)$ is a marked point process  and $G:\R_+\times\R^d \to \R$ a   measurable function, we call a stochastic process $S=(S_t)_{t \ge 0}$ having the representation  
\begin{align}\label{eq:defshotnoise}
  S_t = \sum_{i=1}^{\infty} \ind{T_i \le t} G(t-T_i,U_i), \qquad t \ge 0,
\end{align}
a \emph{shot-noise process}. If the process $\sum_{T_i \le t} U_i,\ t \ge 0$ has independent increments, we call $S$ an \emph{inhomogeneous shot-noise process} and if the
increments are moreover identically distributed $S$ is called \emph{standard shot-noise process}. 
\end{definition}
The \emph{classical shot-noise process} is obtained when $G(t,u)$ does not depend on $u$, see \cite{BondessonSN} for links to the literature on this class.
Time-inhomogeneous L\'evy processes have independent increments and hence may serve  as a useful  class of driving processes; see \cite{JacodShiryaev} for an in-depth study of processes with independent increments and, for example, \cite{Sato,ContTankov} for a guide to the rich literature on L\'evy processes. The interest in driving processes beyond processes with independent increments can be traced back to \cite{Ramakrishnan53,Woollcott73,VSchmidt87} -- only under additional assumptions explicit formulae can be obtained.

Note that absolute convergence of the infinite sum in \eqref{eq:defshotnoise} is implicit in our assumption and needs not be true in general. However, when the stopping times $(T_i)_{i \ge 1}$ have no accumulation point, this will always hold. A precise definition of this technical fact will utilize the relation to random measures and the associated compensators, which we introduce now.

To the marked point process $Z$ we associate an  integer-valued random measure $\mu$ on $\R_+ \times \R^d$ by letting
\begin{align} \label{def:mu}
    \mu_t(A) = \mu([0,t] \times A) :=  \sum_{i \ge 1} \ind{U_i \in A} \ind{T_i \le t}, \quad t \ge 0 
  \end{align}
for any $A \in \cB(\R^d)$. Sometimes we consider the to $Z$ associated process  $Z'=(Z'_t)_{t \ge 0}$ of accumulated jumps given by  $Z'_t = \sum_{i \ge 1} U_i \ind{T_i \le t}$. As usual, we define  $\tilde \Omega = \Omega \times \R_{\ge 0} \times \R^d$, $\tilde \cP=\cP  \otimes \cB(R^d)$, and $\tilde \cO = \cO \otimes \cB(\R^d)$. A $\tilde \cO$-measurable function $W$ on $\tilde \Omega$ is called \emph{optional}. For  an optional function $W$ and a random measure $\mu$ we define 
$$ W*\mu_t = \int_{[0,t] \times \R^d} W(s,x) \mu(ds,dx), \quad t \ge 0,$$
if $\int_{[0,t] \times \R^d} |W(s,x)| \mu(ds,dx)$ is finite, and $W*\mu_t=+\infty$ otherwise. 

From Definition \ref{def:shotnoise}, we obtain that a shot-noise process $S$ has the representation
\[
  S_t =  \int_0^t \int_{\R^d} G(t-s,x) \mu(ds,dx), \quad t \ge 0
\]
and in Lemma \ref{lem:semimart} we show that, if $G$ is absolutely continuous, then $S$ is a semimartingale.

The \emph{compensator} of $\mu$ is the unique, $\bbF$-predictable random measure  $\nu$ such that
$$ \E[W*\mu_\infty] = \E[W * \nu_\infty] $$
 for any non-negative $\tilde \cP$-measurable function $W$ on $\tilde \Omega$, see Theorem II.1.8 in \cite{JacodShiryaev}.

Some properties of the marked point process $Z$ can be determined from the compensator: if $\nu$ is deterministic (i.e.\ does not depend on $\omega$), then $Z'$ has independent increments. Moreover, if the compensator additionally does not depend on time, i.e.\ $\nu(dt,dx)=\nu(dx)dt $, then $Z'$ also has stationary increments. Under the additional assumption that the jumps have an infinitely divisible distribution we obtain the important special case that $Z'$  is a L\'evy process.

\begin{example}[Exponential decay]
An important special case is the well-known case when the decay is exponential. We will later show that this is essentially the only case when $S$ is Markovian. Consider $d=1$, assume that $\nu([0,t],\R)< \infty$ for all $t \ge 0$ and denote $Z'_t = \sum_{T_i \le t} U_i, t \ge 0$.
When $G(t,x) = x e^{-bt}$, we obtain $\partial_t G(t,x) = -b G(t,x)$ and $G(0,x)=x$, such that, by It\^o's formula, 
\[ 
  S_t = \int _0^t - b S_u du + Z'_t.
\] 
Hence, if $Z'$ has independent increments, then $S$ is  a Markov process, in particular, an Ornstein-Uhlenbeck process. 
\end{example}

We give some further useful specifications of shot-noise processes to be used in the following.

\begin{example}\label{Ex Jump Types}
Specific choices of the noise function $G$ lead to processes with independent increments, Markovian, and non-Markovian processes.
\begin{enumeratei}
\item A \emph{jump to a new level} (with $d=1$ and $G(t,x) = x$). Then $Z'=S$ and $S$ has the same properties, as for example independent and stationary increments, such that $S$ is a L\'evy process.
\item We say that $S$ has \emph{power-law decay}  when 
$$G(t,x) = \frac{x}{1+ct}$$
with some $c>0$. This case allows for long-memory effects and heavy clustering, compare \cite{MorenoSerranoStute2011}. In this case, the noise decay is slower than for the exponential case and the effect of the shot persists for longer time in the data.
\item We say that $S$ has \emph{random decay} if the decay parameter is random. This is an interesting extension of the class of Ornstein-Uhlenbeck processes. For example,  let $d=2$ and
$$G(t,(u,v)) = u \exp(-v t).$$
Clearly, jump height and decay size can be dependent, see also \cite{SchmidtStute07}. 
\end{enumeratei}
\end{example} 

Shot-noise processes offer a parsimonious and flexible framework as we illustrate in the following example. In general, shot-noise processes are not necessarily semimartingales: indeed, this is the case if $t\mapsto G(t,x)$ is of infinite variation for all $x$ (or for at least some $x$).

The following result, which is well-known for standard shot-noise processes, gives the conditional characteristic function of $S$. This is a key result to the following applications to credit risk. We give a proof using martingale techniques which is suitable for our setup.

\begin{proposition}\label{lem:SN}
Assume that $S$ is a  shot-noise process, $\nu(dt,dx)$ does not depend on $\omega$, and
 $  \nu([0,T],\R^d) < \infty$. Then, for any $0 \le t \le T$ and $\theta \in \R$, 
\begin{align}\label{eq:fourierSN}
  \E\Big[ e^{i \theta S_T}|\cF_t \Big]  
    &= e^{ i \theta \int_0^t \int_{\R^d} G(T-s,x) \mu(ds,dx)} \cdot
       \exp\bigg( \int_t^T \int_{\R^d}\Big( e^{i \theta G(T-s,x)}-1 \Big) \nu(ds,dx) \bigg).  
\end{align}
\end{proposition}

\begin{proof}
For fixed $T$ and $\theta$, we define 
$$ X_t := \exp\bigg( i \theta \int_0^t \int_{\R^d} G(T-u,x) \mu(du,dx) \bigg),
\quad 0 \le t \le T. $$
By It\^o's formula we obtain that
\begin{align*}
X_t = 1 + \int_0^t X_{s-} \Big( e^{i \theta G(T-s,x)}-1 \Big) \mu(ds,dx), \quad 0 \le t \le T.
\end{align*}
We set $\varphi(t):= \E[X_t]$ for $t \in [0,T]$. Then
\begin{align*}
  \varphi(T)  &= \E\Big[ e^{i \theta S_T}\Big] \\
  &= 1+ \E\bigg[ \int_0^T X_{t-} \int_{\R^d} \Big( e^{i \theta G(T-t,x)}-1 \Big) \nu(dt,dx) + M_T \bigg] \\
  &= 1 + \int_0^T \varphi(t-)  F(dt), \quad 0 \le t \le T, 
\end{align*}
where  $M$ is a martingale and $F(t) = \int_0^t \int_{\R^d} \Big( e^{i \theta G(T-t,x)}-1 \Big) \nu(dt,dx), \ 0 \le t \le T $ is an increasing function with associated measure $F(dt)$. The unique solution of this equation is given by 
\begin{align*}
\varphi(T) &= \exp(F(T)) \\
  &= \exp\bigg( \int_0^T \int_{\R^d} \Big( e^{i \theta G(T-t,x)}-1 \Big) \nu(dt,dx)  \bigg).
\end{align*}
Finally, we observe that for $0 \le t \le T$, 
\begin{align*} 
  \E\Big[ e^{i \theta S_T}|\cF_t \Big] 
    &= e^{ i \theta \int_0^t \int_{\R^d} G(T-s,x) \mu(ds,dx)} \cdot
      \E\bigg[ \exp\bigg( i \theta \int_t^T \int_{\R^d} G(T-s,x) \mu(ds,dx) \bigg)|\cF_t\bigg].  
\end{align*}
As $\mu$ has independent increments by assumption, the conditional expectation is in fact an ordinary expectation which can be computed as above and we obtain the desired result.
 \end{proof}
The first part of \eqref{eq:fourierSN} corresponds to the noise of already occurred shots (at time $t$). The second part denotes the expectation of future jumps in $S$. By the application of iterated conditional expectations, Proposition \ref{lem:SN} also allows to compute the finite-dimensional distributions of $S$.

\begin{example}[The standard shot-noise process]\label{rem:standardSN}
If $Z'$ is a compound Poisson process, then $ \nu(ds,dx)=\lambda F_U(dx)ds$ where $\lambda$ is the arrival rate of the jumps, being i.i.d.\ with distribution $F_U$. The classical proof of the above results uses that the jump times of a Poisson process have the same distribution as order statistics of uniformly distributed random variables, see p. 502~in \cite{RolskiSchmidliSchmidtTeugels99}. 
In this case the proof simplifies to 
\begin{align*}
  \E\Big[ e^{i \theta S_T} \Big] &= \E\Big[\sum_{n \ge 1}\ind{T_n \le T, T_{n+1} > T}e^{i \theta \sum_{j=1}^n G(t-T_i,U_i)} \Big] \\
  &= e^{-\lambda T}\frac{(\lambda T)^n}{n!} \prod_{j=1}^n \frac{1}{T} \int_0^T \int_{R^d} e^{i \theta G(t-s,u)} F_U(du)ds \\
  &= \exp\bigg( - \lambda T + \lambda  \int_0^T \int_{R^d} e^{i \theta G(t-s,u)} F_U(du)ds \bigg) \\
  &= \exp\bigg(  \int_0^T \int_{R^d} (e^{i \theta G(t-s,u)}-1) \lambda F_U(du)ds \bigg).
\end{align*}
A conditional version is obtained in an analogous manner.
\end{example}

\begin{remark}[On the general case] What can be said when $\nu$ is not deterministic? In fact, for the proof we need to compute
\begin{align}\label{temp692}
 \E\bigg[ \exp\bigg( i \theta \int_t^T \int_{\R^d} G(T-u,x) \mu(du,dx) \bigg) | \cF_t \bigg]. 
 \end{align}
For this we need to obtain the \emph{exponential compensator} of $\mu$ given $\cF_t$, i.e.\ the $\cF_t$-measurable random measure $\gamma^t$, such that
\begin{align*}
\eqref{temp692} = \exp\bigg( i \theta \int_t^T \int_{\R^d} G(T-u,x) \gamma^t(du,dx) \bigg).
\end{align*}
Exponential compensators for semimartingales were introduced in \cite{KallsenShiryaev2002} and play an important r\^ole in interest rate theory, compare \cite{CuchieroFontanaGnoatto2016}. We will show later that for affine shot-noise processes we will be able to compute the exponential compensator efficiently, see Example \ref{ex:seaffine} where we study a self-exciting shot-noise process.
\end{remark}

The following result, taken from \cite{Schmidt2014:CAT}, gives sufficient conditions which yield that $S$ is a semimartingale.
\begin{lemma}\label{lem:semimart}
Fix $T>0$ and assume that $G(t,x) = G(0,x) + \int_0^t g(s,x) ds$  for all $0 \le t \le T$ and all $x \in \R^d$. If
\begin{align}\label{Fubinicondition}
   \int_0^T \int_{\R^d} (g(s,x))^2 \nu(ds,dx)< \infty,
\end{align}
$\P$-a.s., then $(S_t)_{0 \le t \le T}$ is a semimartingale.
\end{lemma}
\noindent For the convenience of the reader we repeat the proof of this result.
\begin{proof}
Under condition \eqref{Fubinicondition}, we can apply the stochastic Fubini theorem in the general version given in Theorem IV.65 in \cite{Protter}. Observe that
\begin{align}
 S_t  &= \int_0^t \int_{\R^d} \int_s^t g (u-s,x) du  \, \mu(ds,dx) + \int_0^t \int_{\R^d} G(0,x) \mu(ds,dx) \notag\\
     &= \int_0^t \int_0^s \int_{\R^d}  g (u-s,x) \mu(ds,dx) \, du + \int_0^t \int_{\R^d} G(0,x) \nu(ds,dx) + M_t,  \label{semimartdecomp}
\end{align} 
with a local martingale $M$. This is the semimartingale represenation of $S$ and hence $S$ is a semimartingale.
 \end{proof}
It is possible to generalize this result to the case where $G(t,x) = G(0,x) + \int_0^t g(s,x) dA(s)$ with a process $A$ of finite variation. Here, however, we do not make use of such a level of generality -- see \cite{JacodShiryaev}, Proposition II.2.9.~for details on the choice of $A$. 

Moreover, a characterization of semimartingales when starting from the more general formulation in \eqref{SN:1} is possible using similar
methodologies, see \cite{SchmidtStute07} for an example.

\begin{remark}Having a driver $Z$ which has independent and stationary increments may be a limitation in some applications.
It is straightforward to allow for more general driving processes. 
For example, consider a filtration $\bbG=(\cG_t)_{t \ge 0}$ satisfying the usual 
conditions. Let $\nu$ be a $\cG_0$-measurable random measure
on $[0,T]\times \R^d$ such that for any open set $A$ in $\R^k$,
\begin{align*}
  \P\bigg( \sum_{T_i \in (s,t]} \ind{ X_i \in A }=k \,\Big|\, \cG_s \bigg) = e^{-\nu((s,t]\times A)} \, \frac{(\nu((s,t]\times A))^k}{k!}.
  \end{align*}
  If $X_1,X_2,\dots$ are i.i.d.\ and independent of $\bbG$, then $Z$ is a $\bbG$-\emph{doubly stochastic marked Poisson process}, Intuitively, given  $\cG$,
  $Z$ is a (time-inhomogeneous) Poisson process with $\cG_0$-measurable jumps. This is a so-called
  initial enlargement of filtration, compare \cite{BieleckiJeanblancRutkowski2000} or \cite{JeanblancRutkowski2000} for an introduction into   this field.
Doubly-stochastic marked Poisson processes in credit risk modelling have also been considered in \cite{GasparSlinko08}, however not in a shot-noise setting.
\end{remark}

\begin{example}[An affine self-exciting shot-noise process] \label{ex:seaffine}
Inspired by \cite{ErraisGieseckeGoldberg2010} we consider the two-dimensional affine process $X=(N,\lambda)^\top$ where
\begin{align} \label{affsn}
  d\lambda_t & = \kappa(\theta - \lambda_t) dt + dN_t 
\end{align}
and $N$ is a counting process with intensity $\lambda$. In this case the compensator of $N$ is given by $\nu_N(dt,dx) = \lambda_t \delta_{1}(dx) dt$; $\delta_a$ denoting the Dirac measure at the point $a$. Following \cite{KellerRessel2013}, the process $X$ is a two-dimensional affine process with state space $\N_0 \times \R_{\ge 0}$ when $\theta \ge 0$. Hence its conditional distribution is given in exponential affine form, i.e.
$$ \E[e^{iu X_T}|\cF_t] = \exp\Big(\phi(T-t,u) + \langle \psi(T-t,u), X_t \rangle\Big), $$
for all $u \in \R^2$ and the coefficients $\phi$ and $\psi$ solve the generalized Riccati equations
\begin{align*}
\partial_t \phi(t,u) &= \kappa \theta \psi_2(t,u) \\
\partial_t \psi_1(t,u) &= 0 \\
\partial_t \psi_2(t,u) &= -\kappa \psi_2(t,u) + \exp( \psi_2(t,u) + \psi_1(t,u) )-1
\end{align*}
with the boundary conditions $\phi(0,u)=0$ and $\psi(0,u)=u$ (see Proposition 3.4 in \cite{KellerRessel2013}). Hence $\psi_1(t,u)=u_1$.
Observe that $\lambda$ is a shot-noise process (when $\lambda_0=\theta=0$): the solution of \eqref{affsn} is
\begin{align*}
   \lambda_t = e^{- \kappa t} \lambda_0 + \theta(1-e^{-\kappa t}) + \sum_{i=1}^{N_t} e^{-\kappa (t-T_i)}, 
\end{align*}
where we denoted by $T_1,T_2,\dots$ the jump times of $N$. Hence, for $\lambda_0=\theta=0$, $\lambda$ is an (affine and Markovian) shot-noise process. 
\end{example}

\subsection{The Markov property}

Proposition \ref{lem:SN} allows us to draw a connection to \emph{affine} processes. This processes have been
 studied intensively in the literature because of their high tractability. If $G$ has exponential dependence on time, more precisely $G(t,x)=xe^{-bt}$, then $G(t+s,x) = G(t,x)e^{-bs}$ which is the key to Markovianity. Then,
\begin{align*}
e^{ i \theta \int_0^t \int_{\R^d} G(T-s,x) \mu(ds,dx)} &= e^{ i \theta \int_0^t \int_{\R^d} e^{-b(T-t)}G(t-s,x) \mu(ds,dx)}  \\
&= e^{ i \theta e^{-b(T-t)} \int_0^t \int_{\R^d} G(t-s,x) \mu(ds,dx)} = e^{i\theta e^{-b(T-t)} S_t},
\end{align*}
such that
\begin{align*}
\E\Big[e^{i\theta S_T   } \big| \cF_t  \Big]  
& = \exp\bigg( \int_t^T \int_{\R^d}\Big( e^{i \theta G(T-s,x)}-1 \Big) \nu(ds,dx) \bigg) \cdot e^{i\theta e^{-b(T-t)} S_t} \\[2mm]
&=: \exp(\phi(t,T,\theta) + \psi(t,T,\theta) S_t),
\end{align*}
which is the exponential-affine structure classifying affine processes. While for affine processes $\phi$ and $\psi$
are determined via solutions of generalized Riccati equations, in the shot-noise case we obtain a simpler 
integral-representation. 
Similar in spirit, we obtain that if the Markov property is satisfied, many expectations simplify considerably, as the following result illustrates.
\begin{corollary}
Consider an inhomogeneous shot-noise process $S$ with $G(t,x)=xe^{-bt}$ and $E|X_i|<\infty, i \ge 1$. Then, for $T>t$,
\begin{align*}
 \E[ S_T | \cF_t] &= e^{-b(T-t)} S_t +\E\bigg[ \sum_{T_i \in (t,T]} U_i  e^{-b(T-T_i)}\bigg].
\end{align*}
\end{corollary}

We now focus our attention on the important question of Markovianity of shot-noise processes. Typically, shot-noise processes are not Markovian. Still, from a computational point of 
view Markovianity could be preferable. Proposition \ref{prop:Markovian} provides a clear classification when the decay function  satisfies $G(t,x)=xH(t)$: 
then Markovianity is equivalent to an \emph{exponential decay}. In more general cases one typically looses Markovianity.

\begin{proposition}\label{prop:Markovian}
Consider a standard shot-noise process $S$ where $G(t,x) =x H(t)$ with a c\`adl\`ag function $H:\R_{\ge 0}\to \R$. Assume that there exists an $\epsilon>0$ such that  $(0,\epsilon] \subset H(\R^+)$. 
Then $S$ is Markovian, if and only if there exist $a,b \in \R$ such that
$$ H(t)=ae^{-b t}. $$
\end{proposition}
\begin{proof}
First, consider the case where  the shot-noise process $S$ is Markovian (with respect to the filtration $\bbF$).
For $s>t$, we have that
\begin{align}\label{3.9}
\E[ S_s | \cF_t] &=
\sum_{T_i \le t} U_i H(s-T_i) + \E\bigg[\sum_{T_i \in (t,s]}   U_i H(s-T_i) \,\Big|\, \cF_t\bigg].
\end{align}
 As $Z$ has independent and stationary increments, we obtain that
\begin{align*}
 & \E\bigg[\sum_{T_i \in (t,s]} U_i  H(s-T_i) \,\big | \, \cF_t \bigg]
= \E\bigg[ \sum_{T_i \in (0,s-t]} U_i H(s-t-T_i) \bigg] 
\end{align*}
is a deterministic function (and hence does not depend on $\omega$).
From Markovianity it follows that
$ \E[ S_s | \cF_t] = \E[ S_s | S_t]=:\tilde F(t,s,S_t)$ for all $0 \le s \le t$, where $\tilde F$ is a measurable function. Hence, we obtain the existence of a measurable function $F:\R^+\times\R^+\times\R$, 
such that 
\begin{eqnarray}\label{eqn:F}
 & \sum_{T_i \le t} U_i H(s-T_i) = F(t,s,S_t)
\end{eqnarray}
a.s.~for all $0 \le t \le s$. If $\P(U_1=0)=1$ the claim holds with $a=0$. Otherwise choose
non-zero $u$ such that \eqref{eqn:F} holds with  $U_1,U_2,\dots$ replaced by $u$. W.l.o.g.~consider $u=1$.
In particular, $F(t,s,H(t-T_1))=H(s-T_1)$ holds a.s. As in Remark \ref{rem:standardSN} we condition on $N_t=n$ and obtain that
\begin{align}\label{3.10}
F(t,s,\sum_{i=1}^n  H(t-\eta_i) ) = \sum_{i=1}^n  H(s-\eta_i) = \sum_{i=1}^n F(t,s, H(t-\eta_i))\end{align}
with probability one, where $\eta_i$ are i.i.d.~$U[0,s]$. 
As $(0,\epsilon]\subset H(\R^+)$, 
\begin{align}\label{eq:linearity}
F(t,s,x_1+\dots+x_n) = \sum_{i=1}^n F(t,s,x_i)
\end{align}
for all $x_1,x_2,\dots \R^+$ and $n \ge 1$ 
except for a null-set with respect to the Lebesgue measure.
Note with $h$ being c\`ad so is $F$ in the third coordinate and we obtain that \eqref{eq:linearity} holds for all 
$x_1,x_2,\dots \in \R^+$. 
Hence, $F$ is additive such that 
$F(t,s,x) = F(t,s,1) x$ (see  Theorem 5.2.1 and Theorem 9.4.3 in \cite{Kuczma}) for all $x \in \R^+$.   

Next,  we  exploit 
$$ F(t,s,1) H(t-u) = H(s-u) $$
for all $0 \le u \le t \le s$ to infer properties of $h$. First, $u=0$ gives $F(t,s,1)H(t)=H(s)$ and so $H(0)\not = 0$ because otherwise $H(s)$ would
vanish for all $s \ge 0$ which contradicts $(0,\epsilon]\subset H(\R^+)$. Next, $u=t$ gives $H(s-t) = F(t,s,1) H(0) $ such that
$$ H(s-t) H(t) =  F(t,s,1) H(0) H(t) = H(s) H(0).$$
This in turn yields that $f:= H(t)/H(0)$  satisfies
$$ f(x+y) = f(x) f(y). $$
Then $f$ is additive and measurable and hence continuous. 
The equation is a multiplicative version of Cauchy's equation and hence $f(x) = e^{-bx}$, see Theorem 13.1.4 in \cite{Kuczma} such that
we obtain $H(x)=H(0) e^{-bx}.$

For the converse, note that if $G(t) = ae^{-bt}$, then
$$ \sum_{T_i \le t} U_i G(s-T_i) = G(s-t) \sum_{T_i \le t} U_i G(t-T_i), $$
and hence \eqref{3.9} yields that $S$ is Markovian.
 \end{proof}

\begin{remark} 
 For Markovianity it is necessary that $U_1, U_2, \dots$ are independent \emph{and} identically distributed.
Merely for the sake of the argument, assume that $U_1,U_2 \in \{0,1,2\}$ and $0=U_3=U_4, \dots$. If $t>T_1$ and $S_t=2$ the distribution of
$S_{t+1}$ depends not only on $S_t$ but also on the number of jumps before $t$  and so it is not Markovian.
\end{remark}

\section{The application to financial markets} \label{sec:FM}
Shot-noise processes have been applied to the modelling of stock markets and to the modelling of intensities, which is useful in credit risk and insurance mathematics. Following the works \cite{AltmannSchmidtStute08,SchmidtStute07} and \cite{MorenoSerranoStute2011} we consider the application to the modelling of stocks. The main idea is to extend the Black-Scholes-Merton framework by a shot-noise component.

In this regard, we denote by $X$ the price process of the stock. According to Lemma \ref{lem:semimart}, a shot-noise process where $G$ is absolutely continuous (in time), i.e. $G(t,x) = G(0,x) + \int_0^t g(s,x) ds$, is a semimartingale. In financial markets, semimartingales are inherently linked to absence of arbitrage via the fundamental theorem of asset pricing, see \cite{delbaen-schachermayer-06} for a thorough discussion, such that from now on we focus on shot-noise processes of this form.

In this regard, consider a marked point process $Z=(T_i,X_i)_{i \ge 1}$ with mark space $\R^d$ and a noise function $g:\R_{\ge 0}\times \R^d$, determining the shot-noise component of our model. As additional driver we consider a one-dimensional Brownian motion $W$, which is independent of $Z$, and the volatility parameter $\sigma >0$. As previously, the integer-valued random measure $\mu$ counts the jumps of the marked point process $Z$, see Equation \eqref{def:mu}. Altogether, we assume that
\begin{align}\label{def:X}
X_t = X_0 \exp\Big( \mu t + \sigma W_t- \frac{\sigma^2t}{2} + \int_{0}^t \sum_{T_i \le s}  g(s-T_i,U_i)ds + \sum_{T_i \le t}G(0,U_i) \Big), \ t \ge 0.
\end{align}
To guarantee absence of arbitrage, one has to find an equivalent martingale measure. However, for statistical estimation of the model it is important to have a nice structure of the process under the risk-neutral measure. It turns out that this is not the case for the minimal martingale measure, studied in \cite{SchmidtStute07}. In the following we aim at classifying all martingale measures by a drift condition and give some hints of possible choices of martingale measures suitable for applications.
The first important step will therefore be to classify all equivalent measures.

\subsection{Equivalent measure changes}
In this section we study general  measure changes which apply to different settings of shot-noise processes. We focus on the measure changes for shot-noise processes in this part and leave the independent Brownian motion $W$ aside in our considerations until Section \ref{sec:driftcondition}.
In this regard, we will not be able to work with a general filtration, but have to assume a specific structure, which is done in the following two assumptions.
We consider an initial $\sigma$-field $\cH$. A flexibility of the initial $\sigma$-field can be useful when considering initial enlargements of the filtration, for example when considering doubly-stochastic processes. 
\begin{description}
\item[{\bf (A1)}] $\cF=\cF_{\infty -}$ and $\bbF$ is the smallest filtration for which $\mu$ is
optional and $\cH \subset \cF_0$.
\item[{\bf (A2)}] The measure $\nu$ is absolutely continuous with respect to the Lebesgue-measure, i.e.~there is a kernel, which we
denote again by $\nu(t,dx)$ such that
$$ \nu(dt,dx)=\nu(t,dx)dt. $$
\end{description}
We set
$$ \xi_n:=\inf\Big\{ t\ge 0: \int_0^t\int_{\R^d} (1-\sqrt{Y(s,u)})^2 \nu(s,du) ds \ge n \Big\},  \quad \text{for all }n. $$
A measure $\P'$ is called \emph{absolutely continuous}  with respect to a second measure $\P$, if $\P(A)=0$ implies $\P'(A)=0$ for all $A\in \cF$. In this case we write $\P'\ll \P$. 
An important tool on a filtered probability space is the \emph{density process} $L$. This is the unique martingale $L$, such that for each $t \ge 0$, $L_t$ coincides with the Radon-Nikodym derivative $d\P' |_{\cF_t} / d\P |_{\cF_t}$. Here we denote by 
$\P |_{\cF_t}$ the restriction of $\P$ to $(\Omega,\cF_t)$ (and similarly for $\P'$).

If $\P' \ll \P$ and $\P \ll \P'$ holds, then we call $\P$ and $\P'$ \emph{equivalent} and denote this relation by $\P\sim \P'$. 
\begin{proposition}\label{prop:absoluteMC}
Assume that {\bf (A1)} and {\bf(A2)} hold and $\P'\ll\P$. Then there exists a $\cP\otimes \cB(\R^d)$-measurable non-negative function $Y$
such that the density process $L$ of $\P'$ relative to $\P$ coincides with
\begin{align}\label{eq:Zn}
L^n_t = e^{-\int_0^{t\wedge \xi_n}\int_{\R^d} (Y(s,u)-1) \nu(s,du) ds} \, \prod_{T_i \le t} Y(T_i,U_i)
\end{align}
on $\llbracket 0,\xi_n \rrbracket$ for all $n \ge 1$.
Moreover, $Z$ is a (possibly explosive) marked point process under $\P'$ and its compensator w.r.t.\ $\P'$ is given by
$ Y(t,u) \nu(t,du) dt. $
\end{proposition}

\begin{proof}
We apply Theorem III.5.43 in \cite{JacodShiryaev} and refer to their notation for this proof.  Note that because the compensator of $Z$ is absolutely continuous,
$$ \hat Y_t=\int_{\R^d} Y(t,u) \nu(\{t\},du) = 0 $$  
(compare Equation III.5.2) and therefore $\sigma$ given in Equation III.5.6 satisfies $\sigma=\infty$. Furthermore, the process $H$ given in  Equation III.5.7 satisfies
$$ H_t = \int_0^t\int_{\R^d}(1-\sqrt{Y(s,u)})^2 \nu(s,du) ds. $$
In general, $H$ could  explode, so that following  III.5.9 we consider  $\xi_n=\inf\{t\ge 0: H_t \ge n\}$ and define $N^{\xi_n}$ by
$$ N_t^{\xi_n} := \int_0^{t \wedge \xi_n} (Y-1) \, (\mu(ds,du) - \nu(s,du)ds). $$
Proposition III.5.10  yields that there exists a unique $N$ which coincides with $N^{\xi^n}$ at least on all random intervals $\llbracket 0,\xi_n\rrbracket$, $n\ge 1$.
Theorem III.5.43  yields that under our assumptions the density $L$ coincides with $L^n$ as  inspection of formula III.5.21 shows. This gives our claim.
 \end{proof}

The main tool is the following result which considers the stronger case of equivalent measures.

\begin{theorem} Assume that {\bf (A1)} and {\bf (A2)} hold and $\P\sim\P'$. Then \label{thm:density1}
\begin{align} \label{cond1}
\int_0^t\int_{\R^d} Y(s,u)\nu(s,du)ds < \infty
\end{align}
$\P'$-almost surely for all $t>0$ and the density $L$ is given by
\begin{align}\label{eq:Z}
L_t = e^{-\int_0^{t}\int_{\R^d} (Y(s,u)-1) \nu(s,du) ds} \, \prod_{T_n \le t} Y(T_n,U_n), \quad t \ge 0.
\end{align}
\end{theorem}

\begin{proof}
As $\P$ and $\P'$ are equivalent and we consider only non-explosive marked point processes, $\P'(\lim_{n \to \infty} T_n=\infty)=1$. 
Hence $\int_0^t\int_{\R^d} \nu(s,du)ds < \infty$ for all $t>0$, almost surely with respect to $\P$ and $\P'$.

Also, $\int_0^t\int_{\R^d} Y(s,u)\nu(s,du)ds < \infty$: let $A_t:= \{\omega \in \Omega: \int_0^t\int_{\R^d} Y(s,u)\nu(s,du)ds =\infty\}$ be a set with positive probability. Then, $Z$ vanishes on $A_t$ and so $\P$ is not equivalent to $\P'$ which gives a contradiction. Because $Y\nu$ is non-negative $A_t \subset A_{t+\epsilon}$ for all $\epsilon>0$ and  \eqref{cond1} follows $\P'$-almost surely for all $t>0$. Finally,
note that
\begin{align*}
\int_0^t\int_{\R^d} (1-\sqrt{Y(s,u)})^2 \nu(s,du) ds \le \int_0^t \int_{\R^d} (1+Y) \nu(s,du)ds < \infty.
\end{align*}
Hence the $\xi_n$ in Proposition \ref{prop:absoluteMC} tend to infinity with probability 1. Then \eqref{eq:Zn} together with
Proposition III.5.10 in \cite{JacodShiryaev} gives \eqref{eq:Z}.
 \end{proof}

We have the following important result: the shot-noise property is preserved under an absolutely continuous (and hence also under an equivalent) change of measure.

\begin{corollary}
Assume that {\bf (A1)} and {\bf(A2)} hold and $\P'\ll\P$. If $S$ is a shot-noise process under $\P$, then $S$ is a shot-noise process under $\P'$.
\end{corollary}
\begin{proof}
The result follows immediately from the Definition \ref{def:shotnoise} together with Proposition \ref{prop:absoluteMC}: under $\P'$, the representation \eqref{eq:defshotnoise} of course still holds and by Proposition \ref{prop:absoluteMC} states that $Z=(T_i,U_i)_{i \ge 1}$ is a marked point process under $\P'$.
 \end{proof}
We will see that additional useful properties, like independent increments are not preserved under the change of measure, such that the specific structures of the shot-noise process under both measures can be substantially different.

\subsection{Preserving independent increments}
Recall that we associated with $Z$ the  process $Z'_t=\sum_{T_i \le t} U_i,\ t \ge 0$  which accumulates the jumps of $Z$. 
For tractability reasons one often considers shot-noise processes driven by a marked point process where $Z'$ has independent increments. 
If the increments are moreover stationary and infinitely divisible, the associated process $Z'$ is a L\'evy process. We cover both cases in this section.

\begin{theorem} \label{thm3.4}Assume that $\P\sim\P'$. Let the density process of $\P'$ relative to $\P$ be of the form \eqref{eq:Z}.
\begin{enumerate}
\item If $Z'$ has independent increments under $\P$ and $\P'$, then $Y$ is deterministic.
\item If $Z'$ has independent and stationary increments under $\P$ and $\P'$, then $Y$ is deterministic and does not depend on time.
\end{enumerate}
\end{theorem}
\begin{proof}
$Z'$ is a process with independent increments (PII), if and only if its compensator is deterministic, see \cite{JacodShiryaev}. Hence, if $Z'$ is a PII  under $\P$, then $\nu(\omega,t,dx)=\nu(t,dx)$ is deterministic.
By Proposition \ref{prop:absoluteMC},  $Z'$ has a deterministic compensator under $\P'$ if and only if $Y(\omega,t,u)\nu(t,du)$ is deterministic and hence
$Y(\omega,t,u)=Y(t,u)$ is deterministic.
Stationarity is equivalent to $\nu$ being independent of time and so (ii)  follows analogously.
 \end{proof}

\begin{example}[The Esscher measure] Consider a generic $n$-dimensional stochastic process $X$. Then the Esscher measure  (\cite{Esscher1932}) is given by the density
$$ L_t = \frac{e^{h X_t}}{\E(e^{h X_t})} $$
where $h \in \R^d$ is chosen in such a way  that $Z$ is a martingale.
\cite{EscheSchweizer2005} showed that the Esscher measure preserves the L\'evy property, in a specific context. It is quite immediate that if applied to a model for stock prices driven by shot-noise processes this property will not hold in general. \cite{DassiosJang2003} applied the Esscher measure to Markovian shot-noise processes.
\end{example}

\begin{example}[The minimal martingale measure] The minimal martingale measure as proposed in \cite{FoellmerSchweizer1990} for a certain class of shot-noise processes has been analysed in \cite{SchmidtStute07}. It can be described as follows: consider the special semimartingale $X$ in its semimartingale decomposition $X=A+M$ where $A$ is an increasing process of bounded variation and $M$ is a local martingale. Assume that there exists a process $\ell$ which satisfies
$$ A_t = \int_0^t \ell_s d \langle M \rangle _s. $$
Then the density of the minimal martingale measure with respect to $\P$ is given by
$$ L = \cE \left( \int_0^{\cdot} \ell_{s-} dM_s \right). $$
Here $\cE$ denotes the Doleans-Dade stochastic exponential, i.e.\ $L$ is the solution of $dL_t=L_{t-} \ell_{t-}dM_t$. The minimal martingale measure need not exist in general. From \eqref{def:X}, proceeding as in the proof of Proposition 4.1 in \cite{SchmidtStute07}, we obtain that
\begin{align*}
  \ell_{t-} &= \frac{1}{X_{t-}} \frac{\mu + \sum_{T_i < t} g(t-T_i,U_i) + \int_{\R^d}(e^{G(0,x)}-1) \nu(t,ds)}{ \int_{\R^d}(e^{G(0,x)}-1)^2 \nu(t,ds)}.
\end{align*}
Conditions which ensure that the minimal martingale measure is indeed a probability measure can be found in \cite{SchmidtStute07}.

From Theorem \ref{thm3.4} it is clear, that the minimal martingale measure will not preserve independent increments of $Z'$ - a property which makes this measure less tractable for financial applications. In the following section, we propose an alternative to this approach.
\end{example}

\subsection{The drift condition}\label{sec:driftcondition}
We consider the equivalent measure $\P'\sim \P$ and assume that {\bf (A1)} and {\bf (A2)} hold. Then Theorem \ref{thm:density1} gives the relationship between both measures and $Z$ is again a marked point process under $\P'$. The compensator of $\mu$ under $\P'$ is given by
$$ \nu'(dt,tx)=\nu'(t,dx)dt = \nu(t,dx) Y(t,x) dt. $$
In addition, we consider, as in Equation \eqref{def:X} a Brownian motion $W$, which is (under $\P$) independent of $Z$.
By the equivalent change of measure, there exists a market price of (diffusive) risk $\xi$, such that $W'=W+\int_0^\cdot \xi_s ds$ is a $\P'$-Brownian motion, see \cite{JacodShiryaev}, Theorem III.3.24.  

We assume that discounting takes place via a bank account with short rate $r$ - here $r$ is a progressively measurable process such that $\int_0^t r_s ds < \infty$ $\P$-a.s.\ for all $t \ge 0$.

\begin{theorem}
The equivalent measure $\P'$ is a (local) martingale measure, if 
\begin{align}\label{dc}
r_t &= \mu- \sigma \xi_t + \int_{\R^d} g(t-s,x)\nu'(t,dx) + \int_{\R^d}\Big(e^{G(0,x)}-1\Big) \nu'(t,dx)
\end{align}
$d\P\otimes dt$-almost surely for all $t \ge 0$.
\end{theorem}
\begin{proof}
We first derive the semimartingale representation of $X$. By It\^o's formula and \eqref{semimartdecomp},
\begin{align*}
dX_t &= X_{t-} \bigg(    \mu dt + \sigma dW_t + \int_0^t\int_{\R^d} g(t-s,x)\mu(ds,dx) \bigg) \\
&+ \int_{\R^d}X_{t-}\Big(e^{G(0,x)}-1\Big) \mu(dt,dx).
\end{align*}
As already mentioned, the equivalent change of measure allows to introduce a drift $\xi$ to the Brownian motion, such that
$ W'=W+\int_0^\cdot \xi_s ds $ is a $\P'$-Brownian motion. Compensating $\mu$ with the $\P'$-compensator $\nu'$ and discounting in the usual way gives the result. 
 \end{proof}

It is apparent that typically there will be many solutions of the drift condition. With a view on tractability it is reasonable to impose that the marked point process $Z$ has independent (and possibly stationary) increments under $\P$ and $\P'$. From Theorem  \ref{thm3.4} it follows that this is the case if the function $Y$ is deterministic (and does not depend on time). Then, from Equation \eqref{dc} we obtain the following condition on the drift $\xi$,
\begin{align}
\xi_t &= \sigma^{-1} \Big(\mu -r_t + \int_0^t\int_{\R^d} g(t-s,x)\mu(ds,dx) + \int_{\R^d}\Big(e^{G(0,x)}-1\Big) Y(t,x)\nu(t,dx) \Big).
\end{align}
If there were no jumps, we obtain that $\xi$ coincides with the classical market price of (diffusive) risk, $\sigma^{-1}(\mu-r)$. 
\begin{example}[Independent and stationary increments under both measures] Fix a finite time horizon $T^*$ and assume that $Z'$ has independent and stationary increments, i.e.~$\nu(t,dx)=\lambda F(dx)$ where $F$ is the distribution of $U_1$ and $\lambda>0$ is the arrival rate of the jumps. Assume that $F'$ is equivalent to $F$, i.e.~ $F'(dx)=\eta(x) F(dx)$ and $\lambda'>0$. Then an equivalent change of measure is obtained via $Y(t,x)=\frac{\lambda'}{\lambda}\eta(x)$. In this case, the arrival rate of jumps under $\P'$ is $\lambda'$ and the jumps sizes are again i.i.d.~with distribution $F'$. Assume that $\int e^{G(0,x)}F'(dx)<\infty$ and let $\xi$ be such that 
\begin{align}
\xi_t &= \sigma^{-1} \Big(\mu -r_t + m_1 + \int_0^t\int_{\R^d} g(t-s,x)\mu(ds,dx)  \Big)  
\end{align}
with $m_1:= \int_{\R^d}\Big(e^{G(0,x)}-1\Big) \lambda' F'(dx)$. If furthermore the process
$$ \bigg(\cE \Big(\int_0^t\xi_s dW_s \Big)\bigg)_{0 \le t \le T^*} $$
is a true martingale, then $\P'$ is an equivalent (local) martingale measure. 
\end{example}

\end{document}